\documentclass[preprint,12pt]{article}
\usepackage{amsfonts}
\usepackage{bbm}
\usepackage{amsmath,booktabs,ctable,threeparttable}
\usepackage{amssymb,amsfonts,boxedminipage}
\usepackage{amsthm}
\usepackage{algorithm}
\usepackage{algorithmic}
\usepackage{epsfig,graphicx,picins,picinpar,subfigure,epstopdf}
\usepackage{multirow}
\usepackage{verbatim}
\usepackage{bm}
\usepackage[round,numbers,authoryear]{natbib}
\usepackage[colorlinks,
            linkcolor=red,
            anchorcolor=blue,
            citecolor=green
            ]{hyperref}

\numberwithin{equation}{section}

\newtheorem{thm}{Theorem}[section]

\theoremstyle{definition}
\newtheorem{defn}[thm]{Definition}

\newtheorem{rem}[thm]{Remark}

\makeatletter
 \newcommand{\norm}[1]{\left\Vert#1\right\Vert}
 \newcommand{\abs}[1]{\left\vert#1\right\vert}
 \newcommand{\set}[1]{\left\{#1\right\}}

\newcommand{\E}[1]{\mathbb{E}[#1]}

\newcommand{\R}[1]{\mathbb{R}#1}

\newcommand{\var}[1]{\text{Var}(#1)}

\newcommand{\bvar}[1]{\text{Var}\bigg(#1\bigg)}

\newcommand{\corr}[1]{\text{Corr}(#1)}
\newcommand{\Unif}[1]{\mathbb{U}(#1)}

\newcommand{\floor}[1]{\lfloor#1\rfloor}
\newcommand{\ceil}[1]{\lceil #1 \rceil}
\newcommand{\real}{\mathbb{R}}
\newcommand{\bsx}{\boldsymbol{x}}

\newcommand{\Rmnum}[1]{\expandafter\@slowromancap\romannumeral #1@}

\newcommand{\ce}{\mathcal{E}}
\newcommand{\ci}{\mathcal{I}}
\newcommand{\ck}{\mathcal{K}}
\newcommand{\cp}{\mathcal{P}}
\newcommand{\cx}{\mathcal{X}}

\newcommand{\bsa}{\boldsymbol{a}}
\newcommand{\simiid}{\stackrel{\mathrm{iid}}\sim}
\newcommand{\mrd}{\,\mathrm{d}}
\newcommand{\tint}{\mathcal{T}_{\mathrm{int}}}
\newcommand{\tbdy}{\mathcal{T}_{\mathrm{bdy}}}

\makeatother

\begin{document}
\title{Extensible grids: uniform sampling on a space-filling curve}
\author{Zhijian He\\Tsinghua University
\and
Art B. Owen\\Stanford University}
\date{June 2014}
\maketitle
\begin{abstract}
We study the properties of points in $[0,1]^d$ generated
by applying Hilbert's space-filling curve to uniformly
distributed points in $[0,1]$.
For deterministic sampling we obtain a discrepancy of $O(n^{-1/d})$
for $d\ge2$.
For random stratified sampling, and scrambled van der Corput
points, we get a mean squared error of $O(n^{-1-2/d})$ for
integration of  Lipshitz continuous integrands, when $d\ge3$. 
These rates are the same as one gets by sampling on $d$ dimensional grids and
they show a deterioration with increasing $d$.  The rate for
Lipshitz functions is however best possible at that level of smoothness
and is better than plain IID sampling.
Unlike grids, space-filling curve sampling provides points
at any desired sample size, and the van der
Corput version is extensible in $n$.  Additionally we show
that certain discontinuous functions with infinite variation
in the sense of Hardy and Krause can be integrated with
a mean squared error of $O(n^{-1-1/d})$. It was previously
known only that the rate was $o(n^{-1})$.
Other space-filling curves, such as those 
due to Sierpinski and Peano, also attain these rates, while
upper bounds for the Lebesgue curve are somewhat worse,
as if the dimension were $\log_2(3)$ times as high.

\smallskip

\end{abstract}
\section{Introduction}
A Hilbert curve is a continuous mapping
$H(x)$ from $[0,1]$ to $[0,1]^d$ for $d>1$. It is an
example of a class of space-filling curves,
of which Peano's was first.  
Space-filling curves have long been mathematically intriguing,
but since the 1980s (see \cite{bader2013space})
they have become important
computational tools in computer graphics, in finding near optimal
solutions to the travelling salesman problem, and in PDE solvers
where elements in a multidimensional mesh must be allocated
to a smaller number of processors \citep{zumbusch2003parallel}.
In this paper we look at a quasi-Monte Carlo method that
takes equidistributed points $x_i\in[0,1]$ and
then uses $P_i = H(x_i)\in[0,1]^d$.
The analysis also provides convergence rates for some
functions that are not smooth enough to benefit 
from unrandomized quasi-Monte Carlo sampling.

Our interest in this problem was sparked by
\cite{gerb:chop:2014}, who present
an innovative combination of sequential
Monte Carlo (SMC), quasi-Monte Carlo (QMC),
and Markov chain Monte Carlo (MCMC) as a
method to compete with particle MCMC.
The resulting method is closely related
to the array-RQMC algorithm of
\cite{lecu:leco:tuff:2006}.

The particle algorithms simulate $N$ copies of a Markov
chain through a sequence of time steps
$t=1,\dots,T$. At the end of time step $t$,
chain $n$ is in position  $\bsx_{nt}\in\real^d$, $n=1,\dots,N$.
The computation to advance a chain from $\bsx_{nt}$
to $\bsx_{n,t+1}$ requires a point in $[0,1]^S$,
which may either be uniformly distributed, in Monte Carlo,
or from a low discrepancy ensemble, in quasi-Monte Carlo.
It is possible to  advance all $N$ chains by one time step
using a matrix $U^{(t+1)}\in [0,1]^{N\times (S+1)}$.
The first column of $U^{(t+1)}$ is used to
identify which row of $U^{(t+1)}$ will be
used to advance each of the points $\bsx_{nt}$,
and then the remaining $S$ columns are used to advance
the $N$ chains.
When $d=1$ we can sort $x_{nt}$ into the same order as the
first column of $U^{(t+1)}$.  Then if $U^{(t+1)}$ is a low
discrepancy point set, the starting positions $x_{nt}$ are
equidistributed with respect to the updating variables.

Things become much more difficult when
$\bsx_{nt}\in\real^d$ for $d\ge 2$.
Then it is not straightforward how one should align $\bsx_{nt}$ with
the first column or first several columns of $U^{(t+1)}$.
\cite{gerb:chop:2014}  place a space-filling curve in $\real^d$.
Each point $\bsx_{nt}\in\real^d$ has a coordinate on this
curve, its pre-image in $[0,1]$.  Then $\bsx_{1t}$ to $\bsx_{Nt}$
are sorted in increasing order of those pre-images,
and the $k$'th largest one is aligned with the row of $U^{(t+1)}$
having the $k$'th largest value in column $1$.

They give conditions under which their algorithm estimates 
expectations with a root mean squared error of $o(n^{-1})$. 
Their sequential Monte Carlo scheme
has a provably better rate of convergence than Monte Carlo
or Markov chain Monte Carlo. Array-RQMC 
behaves empirically as if it has a better rate \citep{lecu:leco:tuff:2006}
but as yet there is no proof. In principal one could simulate the
chains through $T$ steps without any remapping by using a 
quasi-Monte Carlo scheme in $[0,1]^{ST}$.  But in such high dimensions
it becomes difficult to construct point sets with meaningfully better
equidistribution than Latin hypercube samples
\citep{mcka:beck:cono:1979} have.

In this paper we examine the simpler related problem 
taking either a QMC or randomized QMC sample within $[0,1]$, 
and applying the Hilbert curve to that sample in order to 
get a quadrature rule in $[0,1]^d$. 
We study the accuracy of quadrature. This strategy has
been used in computer graphics for related purposes.
\cite{rafajlowicz2008equidistributed} applied a two dimensional
space-filling curve to a Kronecker sequence in $[0,1]$ in order
to downsample an image. They report that this strategy allows
them to approximate the Fourier spectra of the images.
\cite{schretter2013direct} report that they can downsample
images with fewer visual artifacts this way than by using a two
dimensional QMC sequence.

Section \ref{hilbert} introduces the Hilbert curve, giving
its important properties.  Section \ref{discrepancy} studies the 
star-discrepancy of the resulting points obtained
as the $d$ dimensional image of one-dimensional low discrepancy points.
We find that the star-discrepancy is $O(n^{-1/d})$, which is very high
considering that quasi-Monte Carlo
rules typically attain the $O(n^{-1+\epsilon})$ rate, where $\epsilon>0$
hides logarithmic factors.
Section \ref{randomization} considers some randomized quasi-Monte
Carlo (RQMC) versions of Hilbert sampling.
The mean squared error converges as $O(n^{-1-2/d})$
for Lipshitz continuous integrands,
and as $O(n^{-1-1/d})$ for certain discontinuous integrands
of infinite variation, studied there. Thus we see a better
than Monte Carlo convergence rate, 
though one that deteriorates with increasing dimension. As a result
we expect the much more complicated proposal of \cite{gerb:chop:2014}  
to have diminishing effectiveness with increasing dimension.
Section \ref{experiment} presents numerical results showing a
close match between mean squared error rates in our theorems and 
observed errors in some example functions. That is, the asymptote
appears to be relevant at small sample sizes.
Section \ref{conclusion} compares the results here to those of other 
methods one might use. 
We find that Hilbert curve quadrature
commonly gives the same convergence rates that one would
see from using grids of $n=m^d$ points in $[0,1]^d$, but makes
those rates available at all integer sizes~$n$. At low smoothness
levels (Lipshitz continuity only) that  poor rate is in fact best possible.

\section{Hilbert Curves}\label{hilbert}

Here we introduce Hilbert's space-filling
curve and some of its properties that we need. 
For more background, there is the monograph
\cite{saga:1994} on space-filling curves, of which Chapter
2 describes Hilbert's curve. \cite{zumbusch2003parallel}
describes multilevel numerical methods, including Chapter 4
on space-filling curves.

Throughout this paper, $d$ is a positive integer, $\lambda_d$ is $d$-dimensional Lebesgue measure, and $\norm{\cdot}$ is the usual Euclidean norm.  For integer $m\ge0$, define $2^{dm}$ intervals
$$I_d^m(k)=\bigg[\frac{k}{2^{dm}},\frac{k+1}{2^{dm}}\bigg],\quad k=0,\dots,2^{dm}-1,$$
and let $\ci_d^m=\set{I_d^m(k)\mid k<2^{dm}}$.
Next, for $\kappa = (k_1,\dots,k_d)$
with $k_j\in\{0,1,\dots,2^m-1\}$
define $2^{dm}$ subcubes of $[0,1]^d$ via
\begin{equation}\label{ElementaryInterval}
  E_d^m(\kappa)= \prod_{j=1}^d\bigg[\frac{k_j}{2^{m}},\frac{k_j+1}{2^{m}}\bigg].
\end{equation}
The set of indices $\kappa$ is $\ck_d^m=\{0,1,\dots,2^m-1\}^d$
and we let $\ce_d^m = \{E_d^m(\kappa)\mid\kappa\in \ck_d^m\}$.
One can find a sequence of mappings $H_m:\ci_d^m\to \ce_d^m$ with the following properties,
\begin{itemize}
  \item Bijection: For $k\neq k'$, $H_m(I_d^m(k))\neq H_m(I_d^m(k'))$.
  \item Adjacency: The two subcubes $H_m(I_d^m(k))$ and $H_m(I_d^m(k+1))$ are adjacent. That is, they have one $(d-1)$-dimensional face in common.
  \item Nesting: If we split $I_d^m(k)$ into the $2^d$ successive subintervals $I_d^{m+1}(k_\ell),\ k_\ell=2^dk+\ell,\  \ell=0,\dots,2^d-1$, then the $H_{m+1}(I_d^{m+1}(k_\ell))$ are subcubes whose union is $H_m(I_d^m(k))$.
\end{itemize}
Figure \ref{CurveShown} illustrates the Hilbert curve construction in dimension $2$.
\begin{figure}
  \centering
  \includegraphics[width=\hsize]{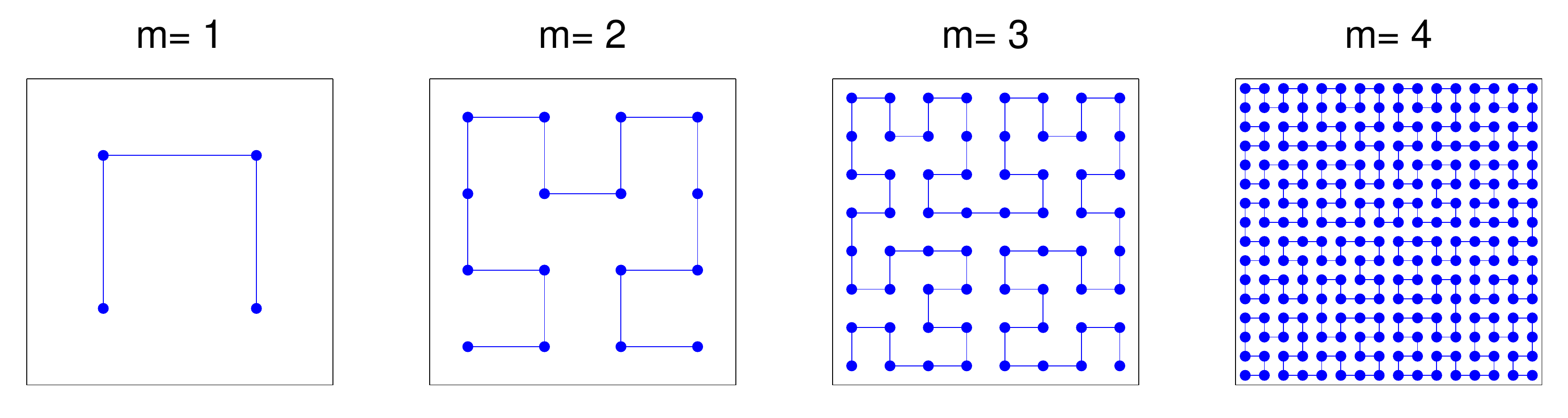}
  \caption{First $4$ stages in the approximation of Hilbert's space-filling curve
}\label{CurveShown}
\end{figure}

The Hilbert curve is defined by $H(x)=\lim_{m\to\infty}H_m(x)$.
The point $x\in[0,1]$ belongs to an infinite sequence
$I_d^m(k_m)$ of intervals which shrink to $x$.
If $x$ does not have a terminating base $2$ representation
then the sequence $I_d^m(k_m)$ is unique and then $H_m(I_d^m(k_m))$
is a unique sequence of subcubes. Points such as
$x=1/4 = 0.01\overline{0}=0.00\overline{1}$ with two 
binary representations nevertheless have 
uniquely defined $H(x)$.
The Hilbert curve passes through every point in $[0,1]^d$.
It is not surjective: there are points $x\ne x'$ with
$H(x)=H(x')$. Indeed, a result of \cite{netto1879beitrag}
shows that no space-filling curve from $[0,1]$
to $[0,1]^d$ for $d>1$ can be bijective.

There is more than one way to define the sequence of
mappings in a Hilbert curve.  But any of those ways
produces a mapping $H$ with these properties:
\begin{itemize}
  \item P(1): $H(I_d^m(k))=H_m(I_d^m(k))$.
  \item P(2): If $A\subset [0,1]$ is measurable, then $\lambda_1(A)=\lambda_d(H(A))$.
  \item P(3): If $x\sim \Unif{[0,1]}$, then $H(x)\sim \Unif{[0,1]^d}$. It admits the change of variables:
      \begin{equation}\label{changeofvariables}
        \mu =\int_{[0,1]^d}f(x)\mrd x = \int_0^1f(H(x))\mrd x.
      \end{equation}
  \item P(4): The function $H(x)$ is H\"older continuous, but nowhere differentiable. More precisely, for any $x,y\in [0,1]$, we have
      \begin{equation}\label{lipschitz}
        \norm{H(x)-H(y)}\leq 2\sqrt{d+3}\abs{x-y}^{1/d},
      \end{equation}
\end{itemize}

The H\"older property P(4) is proved in \cite{zumbusch2003parallel}.
We prove it here too, because the proof is short
and we make extensive use of that result.

\begin{thm}
If $x,y\in[0,1]$ and $H$ is Hilbert's space-filling
curve in dimension $d\ge1$, then
        $\norm{H(x)-H(y)}\leq 2\sqrt{d+3}\abs{x-y}^{1/d}$.
\end{thm}
\begin{proof}
Without loss of generality, $x<y$. Let $m=\floor{-\log_2\abs{x-y}/d}$ so that $2^{-dm}\geq \abs{x-y} >2^{-d(m+1)}$. 
The interval $[x,y]$ is contained within one, or at most two, consecutive intervals $I_d^m(k')$, $I_d^m(k'+1)$ for some $k'<2^{dm}-1$.
As a result, the image $H([x,y])$ lies within $H(I_d^m(k'))\cup H(I_d^m(k'+1))$. By P(1) and the adjacency property of $H_m$, the diameter of $H([x,y])$ is bounded by the diameter of two adjacent subcubes of the form (\ref{ElementaryInterval}), which is 
$2^{-m}\sqrt{d+3}\le2\sqrt{d+3}\abs{x-y}^{1/d}$.
\end{proof}

In the context of numerical integration, the integral $\mu$ in (\ref{changeofvariables}) can be estimated by the following average:
\begin{equation}\label{estimator}
  \hat{\mu}=\frac{1}{n}\sum_{i=1}^nf(H(x_i)),
\end{equation}
where $x_i$'s are carefully chosen quadrature points in $[0,1]$. The space-filling curve reduces a multidimensional integral to a one-dimensional numerical integration problem. It is important to point out that the integrand $f\circ H(x)$ is not of bounded variation even for smooth (but non-trivial) functions $f$.
Bounded variation would have yielded convergence rates of $|\hat\mu-\mu|=O(1/n)$
in any dimension via the Koksma-Hlawka inequality (see Section \ref{discrepancy}).

\section{Star-discrepancy}\label{discrepancy}

Given a sequence $x_1,\dots,x_n$ in $[0,1]$, we can obtain a corresponding sequence $P_1,\dots,P_n$ in $[0,1]^d$ by the 
Hilbert mapping function described above, $P_i=H(x_i)$.

We use the star-discrepancy to measure the uniformity of the resulting sequence $\cp = (P_1,\dots,P_n)$. 
For $\bsa=(a_1,\dots,a_d)\in[0,1]^d$,
let $S=\prod_{i=1}^d [0,a_i)$ be the anchored box $[0,\bsa)$,
and let $A(\cp,S)$ denote the number of points $P_i$ in $S$. 
The signed discrepancy of $\cp$ at $S$ is
$$
\delta(S)=\delta(S;\cp) = 
\frac{A(\mathcal{P},S)}{n}-\lambda_d(S)
$$
and the star-discrepancy of $\cp$ is 
\begin{equation}\label{star-discrepancy}
  D^*_n(\mathcal{P})=
\sup_{\bsa\in [0,1)^d}\abs{ \delta([0,\bsa);\cp)}.
\end{equation}
The significance of the star discrepancy comes from the Koksma-Hlawka
inequality: 
\begin{align}\label{koksmahlawka}
|\hat\mu-\mu|\le D_n^*(\cp)V_{\mathrm{HK}}(f)
\end{align}
where $V_{\mathrm{HK}}(f)$ is the total variation of $f$
in the sense of Hardy and Krause \citep{nied:1992}.

We can trivially get a small $D^*_n(\cp)$ by taking
$x_i$ to be the preimage under $H$ of a low discrepancy
point set in $[0,1]^d$. The Hilbert curve clearly adds
no value for such a construction.
For practical purposes we consider
only $x_i$ generated as low discrepancy points in $[0,1]$.

One such construction is the lattice,
\begin{equation}\label{equallyspacedseq}
  x_i=\frac{i-1}{n}.
\end{equation}
The lattice~\eqref{equallyspacedseq} has star discrepancy
$1/n$ and the lowest possible star discrepancy \citep{nied:1992}
for $n$ points in $[0,1]$ is $1/(2n)$
attained via $x_i = (i-1/2)/n$ for $i=1,\dots,n$.

Another such construction is the van der Corput sequence \citep{corput1935}. In van der Corput sampling of $[0,1)$, the integer $i\geq 0$ is written in integer base $b\geq 2$ as $i=\sum_{k= 1}^\infty d_kb^{k-1}$ for $d_k=d_k(i)\in\set{0,1,\dots,b-1}$. Then $i$ is mapped to
\begin{equation}\label{vandercorput}
  x_i =\sum_{k= 1}^\infty d_kb^{-k}.
\end{equation}
The star-discrepancy of the van der Corput sequence is $O(n^{-1}\log(n))$. The van der Corput sequence can be extended one point at a time, while the lattice sequence is not extensible except by doubling the sample size. Figures \ref{latticemapping} and \ref{vandercorputmapping} show the Hilbert mappings from lattice sequence and van der Corput sequence, respectively. For $n=b^m$, the van der Corput sequence in base $b$ is a permutation of the lattice sequence.

In Theorem~\ref{ThmDiscrepancy}
we bound the star-discrepancy of 
stratified $x_i$ like \eqref{equallyspacedseq}.
Figure \ref{partition} shows some of these strata for small $n$.

\begin{figure}[t!]
  \centering
\includegraphics[width=\hsize]{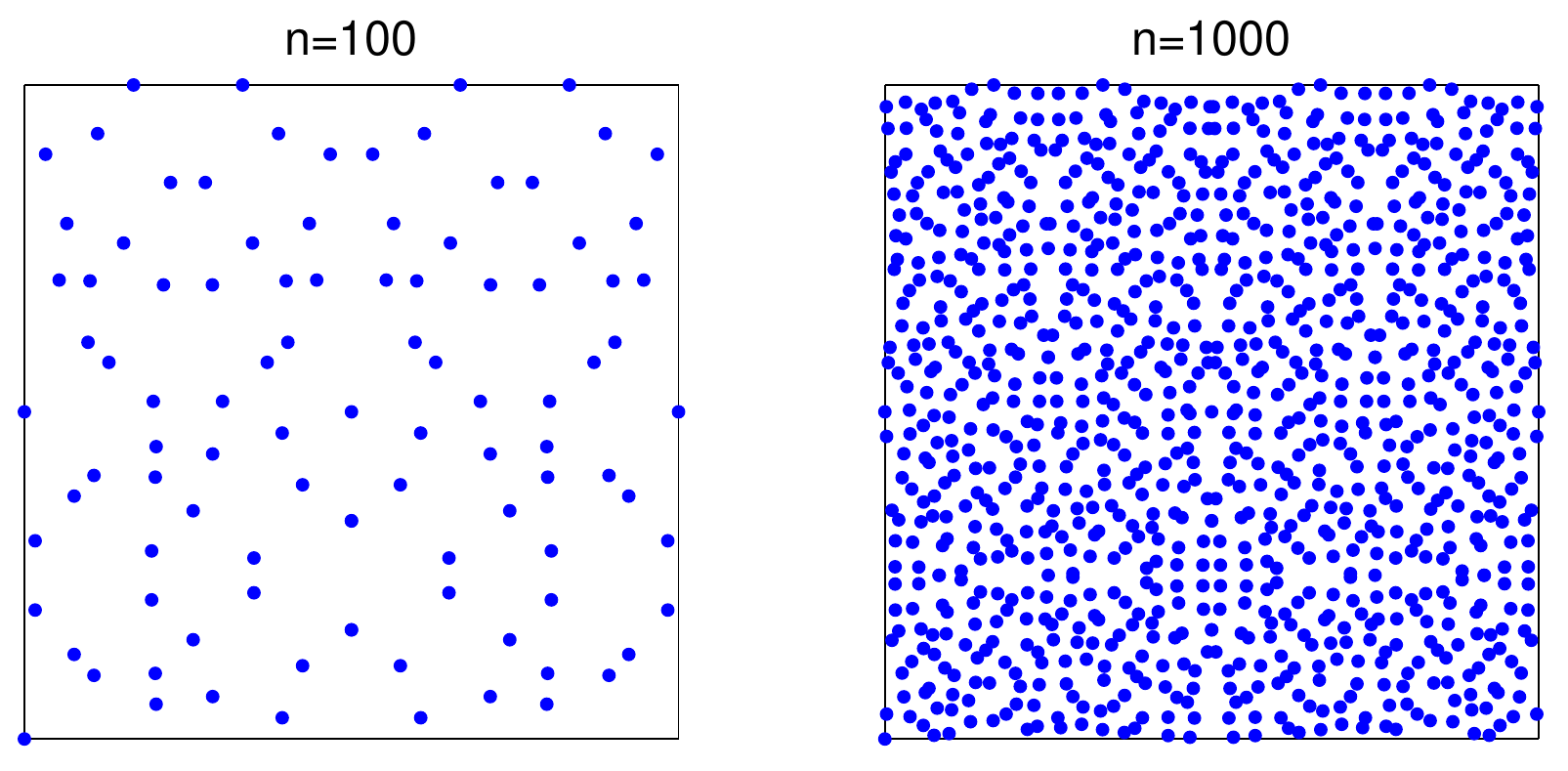}
\caption{Hilbert mappings
of $(i-1)/n$ to $[0,1]^2$ for $i=1,\dots,n$
where $n\in\{100,1000\}$.}\label{latticemapping}
\end{figure}
\begin{figure}
  \centering
\includegraphics[width=\hsize]{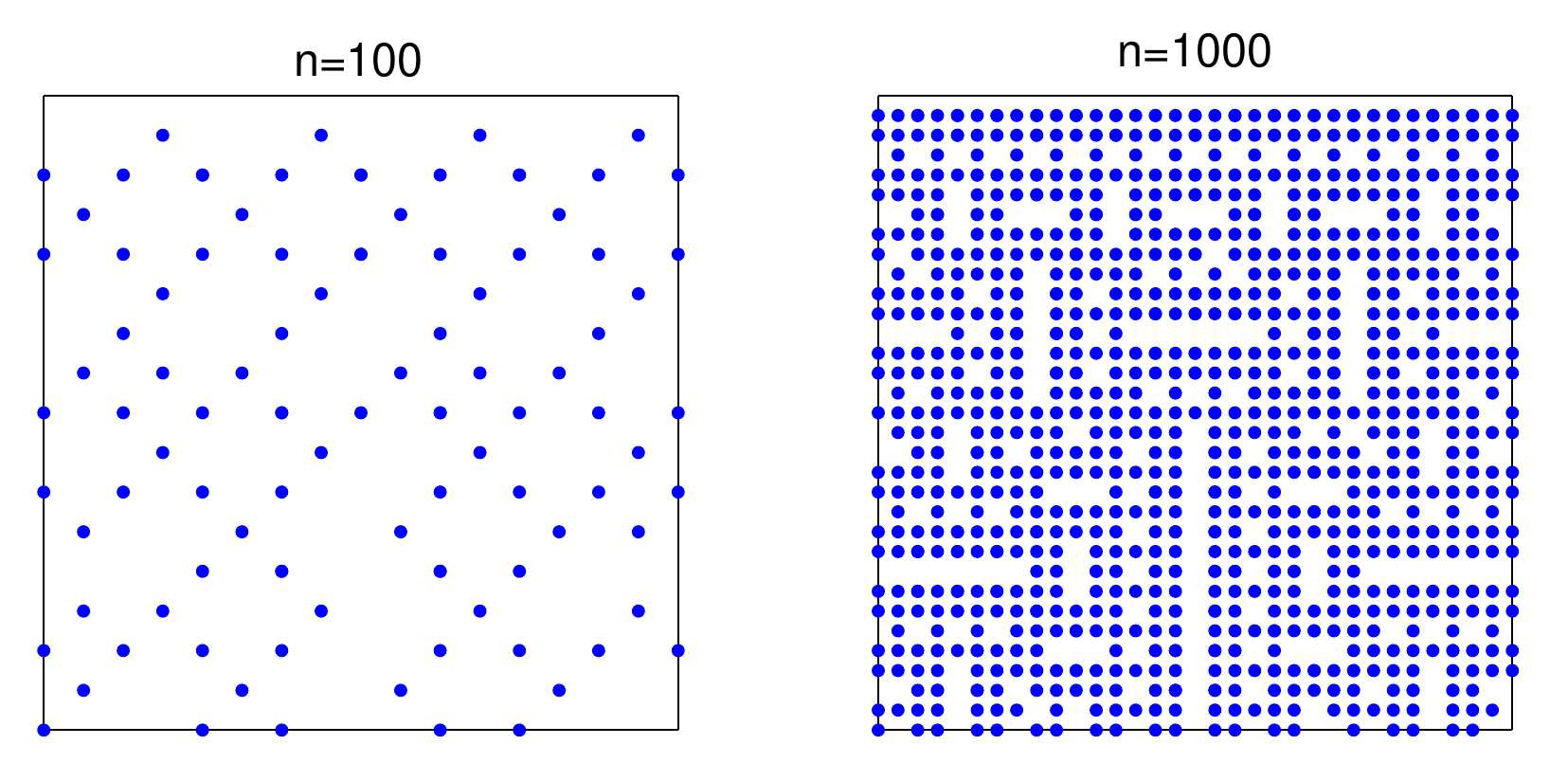}
\caption{Hilbert mappings of the first $n$ van der Corput points
in base $2$ to $[0,1]^2$, for $n\in\{100,1000\}$.
}\label{vandercorputmapping}
\end{figure}
\begin{figure}
  \centering
\includegraphics[width=\hsize]{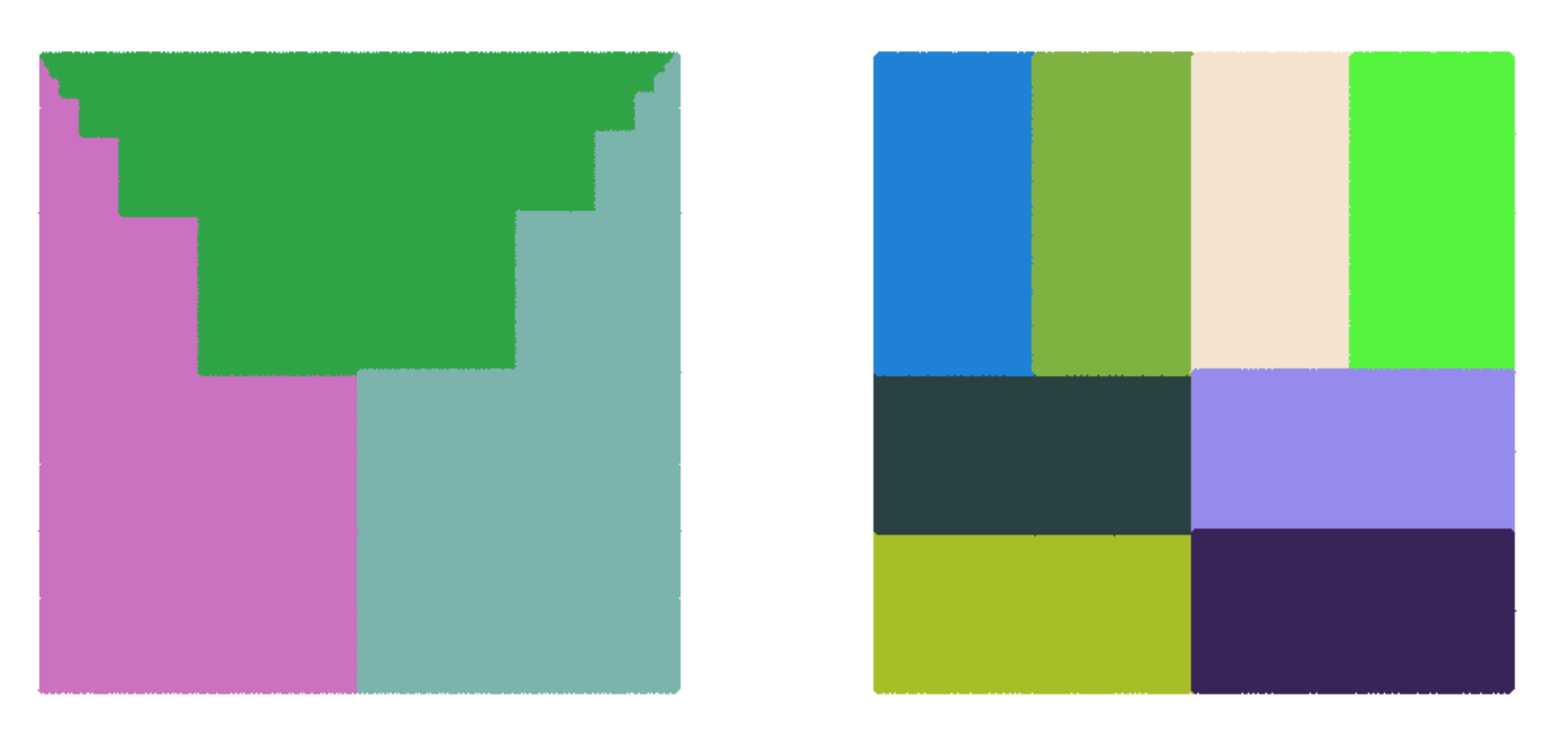}
\caption{Uniform partitions of $[0,1]^2$ by the mapping $H$ for $n=3, 8$.}
\label{partition}
\end{figure}

\begin{thm}\label{ThmDiscrepancy}
Let $x_1,\dots,x_n\in[0,1]$
and let $\mathcal{P}=(P_1,\dots,P_n)$ where $P_i=H(x_i)\in[0,1]^d$.
If each interval $I_k = [(k-1)/n,k/n)$, for $k=1,\dots,n$ contains precisely
one of the $x_i$, then
\begin{equation}\label{equalsq}
  D^*_n(\mathcal{P})\leq 4d\sqrt{d+3}n^{-1/d}+O(n^{-2/d}).
\end{equation}
\end{thm}
\begin{proof}
Choose any $\bsa\in[0,1]^d$ and let $S=[0,\bsa)$.
Next, define $E_k = H(I_k)$ for $k=1,\dots,n$ and
adjoin $E_{n+1} = H( \{1\})$. By additivity of
signed discrepancy,
$$\delta(S;\cp) 
= \frac1n\sum_{k=1}^{n+1} \delta(S\cap E_k;\cp)
= \frac1n\sum_{k=1}^{n} \delta(S\cap E_k;\cp),$$
because $S\cap E_{n+1}$ has volume $0$ and has no points of $\cp$.
From here on, we restrict attention to $E_k$ for $k=1,\dots,n$.
If $E_k\cap S=\emptyset$, then $\delta(S\cap E_k;\cp)=0$.
By the measure preserving property of $H$, $\lambda_d(E_k)=1/n$,
and so if $E_k\subseteq S$, then $\delta(S\cap E_k;\cp)=0$.
Otherwise $-1/n\le\delta(S\cap E_k;\cp)\le 1/n$.
Let $B$ be the number  of `boundary' $E_k$,
which intersect both $S$ and $S^c=[0,1]^d\setminus{S}$.
Then $|\delta(S;\cp)|\le B/n$, and we turn to bounding $B$.

Let $r_k$ be the diameter of $E_k$.
By (\ref{lipschitz}), $r_k\leq \varepsilon\equiv 2\sqrt{d+3}n^{-1/d}$.
Define $S_+ = \prod_{j=1}^d[0,\min(a_j+\varepsilon,1))$
and $S_- = \prod_{j=1}^d[0,\max(a_j-\varepsilon,0))$.
If $E_k$ intersects $S$ and $S^c$, then $E_k\subset S_+\setminus{S_-}$.
Because the $E_k$ are disjoint with volume $1/n$,
$$B\le 
n\lambda_d(S_+\setminus{S_-})
\le 2n\biggl(\,\prod_{j=1}^d(a_j+\varepsilon)-\prod_{j=1}^da_j\biggr)
\le 2n(d\varepsilon+O(\varepsilon^2)).
$$
Thus $|\delta(S;\cp)|\le 2d\varepsilon +O(\varepsilon^2)
=4d\sqrt{d+3}n^{-1/d}+O(n^{-2/d})$, and since $S$ was
any anchored box, the result now follows.
\end{proof}

In Theorem~\ref{ThmDiscrepancyVDC}
we apply Theorem~\ref{ThmDiscrepancy}
to get a bound for star discrepancy
of the van der Corput sequence when $d>1$.
The case $d=1$ is well known and has
star discrepancy $O(\log(n)/n)$.
\begin{thm}\label{ThmDiscrepancyVDC}
For integer base $b\ge 2$ and $n\ge 1$,
let $x_1,\dots,x_n\in[0,1)$
be defined by the van der Corput 
mapping~\eqref{vandercorput}
and let $\mathcal{P}=(P_1,\dots,P_n)$ where $P_i=H(x_i)\in[0,1]^d$.
Then, for $d>1$,
\begin{equation}\label{discvdc}
  D^*_n(\mathcal{P})\leq 
\frac{4(b-1)\sqrt{d+3}}{1-b^{-(d-1)/d}}\,
n^{-1/d}+O(n^{-2/d}\log(n)).
\end{equation}
\end{thm}
\begin{proof}
We begin by writing $n=\sum_{j=0}^ka_jb^j$
where $a_j\in\{0,1,\dots,b-1\}$ and $a_k>0$.
The $x_i$ can be partitioned into disjoint
sets $\cx_{j\ell}$ of length $b^j$ for $\ell=1,\dots,a_j$.
Each of these sets satisfies the conditions of 
Theorem~\ref{ThmDiscrepancy}. Let $\cp_{j\ell}$ be
the image of the points in $\cx_{j\ell}$ under $H$.

Now let $S$ be any anchored box $[0,\bsa)\subset[0,1]^d$.
By additivity of local discrepancy over samples,
$n\delta(S;\cp) = \sum_{j=0}^k\sum_{\ell=1}^{a_j}b^j\delta(S;\cp_{j\ell})$.
Therefore
\begin{align*}
n|\delta(S;\cp)|
& \le \sum_{j=0}^k\sum_{\ell=1}^{a_j}
b^j\bigl[ 4d\sqrt{d+3}b^{-j/d}+O(b^{-2j/d})\bigr]\\
& \le C\sum_{j=0}^k
b^j\bigl[b^{-j/d}+O(b^{-2j/d})\bigr]
\end{align*}
for $C=(b-1)4d\sqrt{d+3}$.
Now for $d>1$,
$$\sum_{j=0}^k(b^{1-1/d})^j\le
\sum_{j=-\infty}^k(b^{1-1/d})^j=\frac{(b^{1-1/d})^k}{1-b^{-1+1/d}}
\le\frac{n^{1-1/d}}{1-b^{-1+1/d}}
$$
and $\sum_{j=0}^k(b^{(1-2/d)})^j\le kb^{k(1-2/d)}
=O(n^{1-2/d}\log(n))$.
\end{proof}


Theorems \ref{ThmDiscrepancy} and \ref{ThmDiscrepancyVDC} show that the star-discrepancy is $O(n^{-1/d})$ for the two sequences. 
Thus the estimate (\ref{estimator}) has a worse upper bound than ordinary QMC if the integrand is of bounded variation in the sense of Hardy and Krause. 

\section{Randomization}\label{randomization}
In this section, we study the variance resulting from randomized samples along 
the Hilbert curve. We get convergence rates for Lipschitz continuous functions.

We also study discontinuous functions of the form
$f(x) = g(x)1_\Omega(x)$ where the set $\Omega\subset[0,1]^d$
has a boundary that admits ($d-1$)-dimensional Minkowski content (defined below).
Functions of this type typically have infinite variation
in the sense of Hardy and Krause \citep{variation} unless the set $\Omega$
is an axis parallel box (or finite union of such).
Infinite variation renders the Koksma-Hlawka inequality~\eqref{koksmahlawka}
useless. We do know that if $f\in L^2[0,1]^d$
then $f\circ H\in L^2[0,1]$ and scrambled net quadrature
on $[0,1]$ for $f\circ H$ will have a mean squared error $o(n^{-1})$.
Here we find a rate.

\subsection{Randomized Lattice Sequence}
We randomize the lattice points in (\ref{equallyspacedseq}) by performing a random shift in each subinterval, that is
\begin{equation}\label{randomshift}
  x_i = \frac{i-1+\Delta_i}{n},\quad \text{with}\quad\Delta_i\simiid \Unif{[0,1]},
\quad i=1,\dots,n.
\end{equation}
As a result, $x_i\sim \Unif{I_i}$ independently for $I_i=[\frac{i-1}{n},\frac{i}{n}]$. Let $\Delta=(\Delta_1,\dots,\Delta_n)$. A randomized version of (\ref{estimator}) is given by
\begin{equation}\label{randomEst}
  \hat{\mu}(\Delta) = \frac{1}{n}\sum_{i=1}^nf(H(x_i)).
\end{equation}
First, we need some definitions.

\begin{defn}\label{defnLipschitz}
For a function $f(x)$ defined on $[0,1]^d$, if there exists a constant $M$ such that $$\abs{f(x)-f(y)}\leq M\norm{x-y}$$ for any $x,y\in [0,1]^d$, then $f(x)$ is said to be Lipschitz continuous.
\end{defn}
\begin{defn}\label{defnMinkowski}
For a set $\Omega\subset [0,1]^d$, define
\begin{equation*}
  \mathcal{M}(\partial \Omega) = \lim_{\epsilon\downarrow 0}\frac{\lambda_d((\partial \Omega)_\epsilon)}{2\epsilon},
\end{equation*}
where $ (\partial \Omega)_\epsilon= \set{X\in \R^d\mid\text{dist}(x,\partial \Omega)\leq \epsilon}.$
If $\mathcal{M}(\partial \Omega)$ exists and is finite, then $\partial \Omega$ is said to admit ($d-1$)-dimensional Minkowski content.
\end{defn}
\begin{thm}\label{ThmLatticeVariance}
The estimate $\hat\mu(\Delta)$ from \eqref{randomEst}
is unbiased for any $f\in L^2([0,1]^d)$. If $f$ is Lipschitz continuous, then
\begin{equation}\label{boundofsmoothfun}
  \var{\hat{\mu}(\Delta)}=O(n^{-1-2/d}).
\end{equation}
If $f(x)=g(x)1_{\Omega}(x)$ where $h$ is Lipschitz continuous and $\partial\Omega$ admits ($d-1$)-dimensional Minkowski content, then
\begin{equation}\label{boundofdiscfun}
  \var{\hat{\mu}(\Delta)}=O(n^{-1-1/d}).
\end{equation}
\end{thm}
\begin{proof}
Let $E_i=H(I_i)$. Because $x_i\sim \Unif{I_i}$, we have $H(x_i)\sim \Unif{E_i}$.
Moreover $\lambda_d(E_i)=1/n$. Thus 
$\E{\hat{\mu}(\Delta)}$ equals
\begin{equation*}
\frac{1}{n}\sum_{i=1}^n\E{f(H(x_i))}=\frac{1}{n}\sum_{i=1}^n\bigg(n\int_{H(I_i)}f(x)\mrd x\bigg)=\int_{[0,1]^d}f(x)\mrd x=\mu,
\end{equation*}
and so  (\ref{randomEst}) is unbiased.

Let $f$ be a Lipschitz continuous function, and let $M$ be
the constant from Definition \ref{defnLipschitz}.
For any $x,y\in E_i$, we have $\abs{f(x)-f(y)}\leq Mr_i$, where $r_i$ is the diameter of $E_i$. As in the proof of Theorem \ref{ThmDiscrepancy}, 
$r_i\le \varepsilon\equiv 2\sqrt{d+3}n^{-1/d}$, and
so $\abs{f(x)-f(y)}\leq 2M\sqrt{d+3}n^{-1/d}$. 
It follows that 
$$\abs{f(H(x_i))-\E{f(H(x_i))}}\leq M\varepsilon = 2M\sqrt{d+3}n^{-1/d},
\quad i=1,\dots,n.$$
Now, since $H(x_i)$'s are independent,
\begin{align*}
  \var{\hat{\mu}(\Delta)} 
  &\leq \frac1{n^2}\sum_{i=1}^n 4M^2(d+3)n^{-2/d}=4M^2(d+3)n^{-1-2/d}
\end{align*}
establishing~\eqref{boundofsmoothfun}.

Next consider $f(x)=g(x)1_{\Omega}(x)$. Let 
$\tint=\set{1\leq i\leq n\mid E_i\subset \Omega}$ and
$\tbdy=\set{1\leq i\leq n\mid E_i\cap\Omega\neq \emptyset}\backslash\tint$. 
These are, respectively, the collections of $E_i$ that are interior
to $\Omega$, and at the boundary  of $\Omega$.
Then
\begin{equation*}
  \hat{\mu}(\Delta) = \frac{1}{n}\sum_{i\in \tint}g(H(x_i))+
\frac{1}{n}\sum_{i\in \tbdy}g(H(x_i))1_{\Omega}(H(x_i))=
\hat{\mu}_{\mathrm{int}}+\hat{\mu}_{\mathrm{bdy}}.
\end{equation*}
Since $g(x)$ is Lipschitz continuous, $\var{\hat{\mu}_{\mathrm{int}}}=O(n^{-1-2/d})$
by the reasoning above. Also, there exists a constant $D$ 
with $\abs{g(x)}\leq D$ for all $x\in [0,1]^d$. Thus
\begin{align}
  \var{\hat{\mu}_{\mathrm{bdy}}} &= \frac1{n^2}\sum_{i\in \tbdy}\var{g(H(x_i))1_{\Omega}(H(x_i))}
\leq \frac{D^2\abs{\tbdy}}{n^2}\label{var2}.
\end{align}
Recall that $\partial\Omega$ admits ($d-1$)-dimensional Minkowski content. It follows from Definition \ref{defnMinkowski} that
\begin{equation*}
  \mathcal{M}(\partial \Omega) = \lim_{\epsilon\downarrow 0}\frac{\lambda_d((\partial \Omega)_\epsilon)}{2\epsilon}<\infty.
\end{equation*}
Thus for any fixed $\delta>2$, there exists $\epsilon_0>0$ such that $\lambda_d((\partial \Omega)_\epsilon)<\delta\mathcal{M}(\partial \Omega)\epsilon$ whenever $\epsilon<\epsilon_0$. We can assume that $n>(2\sqrt{d+3}/\epsilon_0)^{d}$. Then $r_i\leq \varepsilon<\epsilon_0$. Notice that $\bigcup_{i\in\tbdy}E_i\subset(\partial \Omega)_\varepsilon$. We thus arrive at $$
\abs{\tbdy}\leq\frac{\lambda_d((\partial \Omega)_\varepsilon)}{\lambda_d(E_i)}\leq \frac{\delta\mathcal{M}(\partial \Omega)\varepsilon}{n^{-1}}= 2\sqrt{d+3}\delta\mathcal{M}(\partial \Omega)n^{1-1/d}.
$$
Now by (\ref{var2}), we have $\var{\hat{\mu}_{\mathrm{bdy}}}=O(n^{-1-1/d})$. 
Finally, from $\var{\hat{\mu}(\Delta)}\leq 
(\sqrt{\var{\hat{\mu}_{\mathrm{int}}}}+\sqrt{\var{\hat{\mu}_{\mathrm{bdy}}}})^2$, 
we obtain $\var{\hat{\mu}(\Delta)}=O(n^{-1-1/d})$.
\end{proof}

\begin{rem}\label{rem1}
If $\Omega$ is a convex set, then it is easy to see that $\partial \Omega$ admits ($d-1$)-dimensional Minkowski content.
Moreover, $\mathcal{M}(\partial \Omega)\leq 2d$ as the outer surface area of a convex set in $[0,1]^d$
is bounded by the surface area of the unit cube $[0,1]^d$, which is $2d$.
Generally, \cite{Ambrosio2008} show that if $\Omega$ has Lipschitz boundary, then $\partial \Omega$ admits ($d-1$)-dimensional Minkowski content. In their terminology, a set $\Omega$ is said to have Lipschitz boundary if for every boundary point $a$ there exists a neighborhood $A$ of $a$, a rotation $R$ in $\R^d$ and a Lipschitz function $f: \R^{d-1}\to \R$ such that $R(\Omega\cap A)=\set{(x,y)\in (\R^{d-1}\times \R)\cap R(A)|y\geq f(x)}$. In other words, $\Omega\cap A$ is the epigraph of a Lipschitz function.
\end{rem}

\begin{rem}
The convergence rate \eqref{boundofdiscfun}
for $f(x) = g(x)1_\Omega(x)$ extends to functions
$f(x) = g_0(x) + \sum_{j=1}^Jg_j(x)1_{\Omega_j}(x)$
where all of the $g_j$ are Lipschitz continuous and all of the $\Omega_j$
have boundaries with finite Minkowski content.
\end{rem}

\subsection{Randomized van der Corput Sequence}
For the van der Corput sequence, we apply  the nested uniform digit scrambling of \cite{owen1995}. Let $a_1,\dots,a_n$ be the first $n$ points of van der Corput sequence in base $b$. We may write $a_i$ in base $b$ expansion $a_i = \sum_{j=1}^{\infty} a_{ij}b^{-j},$ where $0\leq a_{ij}<b$ for all $i,j$. The scrambled version of $a_1,\dots,a_n$ is a sequence $x_1,\dots,x_n$ written as $x_i = \sum_{j=1}^{\infty}x_{ij}b^{-j},$ where $x_{ij}$ are defined in terms of random permutations of the $a_{ij}$. The permutation applied to $a_{ij}$ depends on the values of $a_{ih}$ for $h<j$. Specifically $x_{i1} = \pi(a_{i1}),\ x_{i2} = \pi_{a_{i1}}(a_{i2}),\ x_{i3} = \pi_{a_{i1}a_{i2}}(a_{i3})$, and generally
$$
  x_{ij} = \pi_{a_{i1}a_{i2}\cdots  a_{ij-1} }(a_{ij}).
$$
Each permutation $\pi_\bullet$ is uniformly distributed over the $b!$ permutations of $\set{0,1,\dots,b-1}$, and the permutations are mutually independent. Let $\Pi$ be the collection of all the permutations involved in the scrambling scheme. The randomized version of (\ref{estimator}) becomes
\begin{equation}\label{scrambledEst}
  \hat{\mu}(\Pi) = \frac{1}{n}\sum_{i=1}^nf(H(x_i)).
\end{equation}
\cite{owen1995} shows that each $x_i$ is uniformly distributed on $[0,1]$. Thus the estimate (\ref{scrambledEst}) is unbiased. Moreover, if $n=b^m$ for some nonnegative $m$, then we can reorder the data values in scrambled sequence such that $x_i\sim \Unif{[\frac{i-1}{n},\frac{i}{n}]}$ independently for $i = 1,\dots,b^m$.
In this case, the scrambled van der Corput sequence is the same as the randomized lattice sequence. Thus the estimate (\ref{scrambledEst}) has the same variance shown in Theorem \ref{ThmLatticeVariance}. For an arbitrary sample size $n$, we can find the associated rates by exploiting the properties of van der Corput sequences.
\begin{thm}\label{ThmvanderCorputVariance}
The estimate $\hat\mu(\Pi)$ of (\ref{scrambledEst}) is unbiased for any $f\in L^2([0,1]^d)$. If $f$ is Lipschitz continuous, then
\begin{equation}\label{boundofsmoothfun2}
  \var{\hat{\mu}(\Pi)}=
\begin{cases}
O(n^{-1-2/d}), & d\ge 3\\
O(n^{-2}\log(n)^2), & d= 2\\
O(n^{-2}), & d= 1.
\end{cases}
\end{equation}
If $f(x)=g(x)1_{\Omega}(x)$ where $g(x)$ is Lipschitz continuous and $\partial\Omega$ admits ($d-1$)-dimensional Minkowski content, then
\begin{equation}\label{boundofdiscfun2}
  \var{\hat{\mu}(\Pi)}=
\begin{cases}
O(n^{-1-1/d}), & d\ge2\\
O(n^{-2}\log(n)^2), & d=1.
\end{cases}
\end{equation}
\end{thm}
\begin{proof}
As in the proof of Theorem~\ref{ThmDiscrepancyVDC} we may
write $n = \sum_{j=0}^ka_jb^j$ with $a_j\in\{0,1,\dots,b-1\}$
where $a_k>0$, and split the points into $\sum_{j=0}^ka_j$
non-overlapping randomized van der Corput sequences,
of which $a_j$ have sample size $b^j$.
Theorem~\ref{ThmLatticeVariance} gives variance bounds
of the form $Cn^{-1-\alpha}$ where $\alpha=2/d$
when $f$ is Lipschitz continuous and $\alpha=1/d$
when $f(x)=g(x)1_\Omega(x)$ for Lipschitz continuous
$g$ and $\partial\Omega$ with bounded Minkowski content.

In either case, 
write $n\hat\mu(\Pi) = \sum_{j=0}^k\sum_{\ell=1}^{a_j}b^j\hat\mu_{j\ell}$, for
$\hat\mu_{j\ell} = \hat\mu_{j\ell}(\Pi)$. Then an elementary inequality
based on $|\corr{\hat\mu_{j\ell},\hat\mu_{j'\ell'}}|\le1$ yields
\begin{align*}
\var{\hat\mu(\Pi)}^{1/2} &
\le C^{1/2}\sum_{j=0}^k\sum_{\ell=0}^{a_j} \frac{b^{j}}{n}\var{\hat\mu_{j\ell}}^{1/2}
\le \frac{(b-1)C^{1/2}}n\sum_{j=0}^kb^{j(1-\alpha)/2}.
\end{align*}
Now for $\alpha<1$
$$\sum_{j=0}^kb^{j(1-\alpha)/2}
\le \frac{b^{k(1-\alpha)/2}}{1-b^{(\alpha-1)/2}}
\le \frac{n^{(1-\alpha)/2}}{1-b^{(\alpha-1)/2}}
$$
using $b^{k}\le n$.
Then 
$\var{\hat\mu(\Pi)} = O( n^{-1-\alpha})$.
That is, for $\alpha<1$, the van der Corput
construction inherits the rate of the stratified one.

For $\alpha=1$
we have $\sum_{j=0}^kb^{j(1-\alpha)/2}= k+1=O(\log(n))$
and then $\var{\hat\mu(\Pi)} = O( n^{-2}\log(n)^2 )$.
For $\alpha>1$
we have $\sum_{j=0}^kb^{j(1-\alpha)/2}=O(1)$
and then $\var{\hat\mu(\Pi)} = O( n^{-2})$.

Equation~\eqref{boundofsmoothfun2} now follows because
$\alpha<1$ for $d\ge3$, $\alpha=1$ for $d=2$ and $\alpha>1$
for $d=1$.  Similarly,~\eqref {boundofdiscfun2} follows because
$\alpha<1$ for $d\ge2$ and $\alpha=1$ for $d=1$.
\end{proof}

\begin{figure*}
  \centering
\includegraphics[width=\hsize]{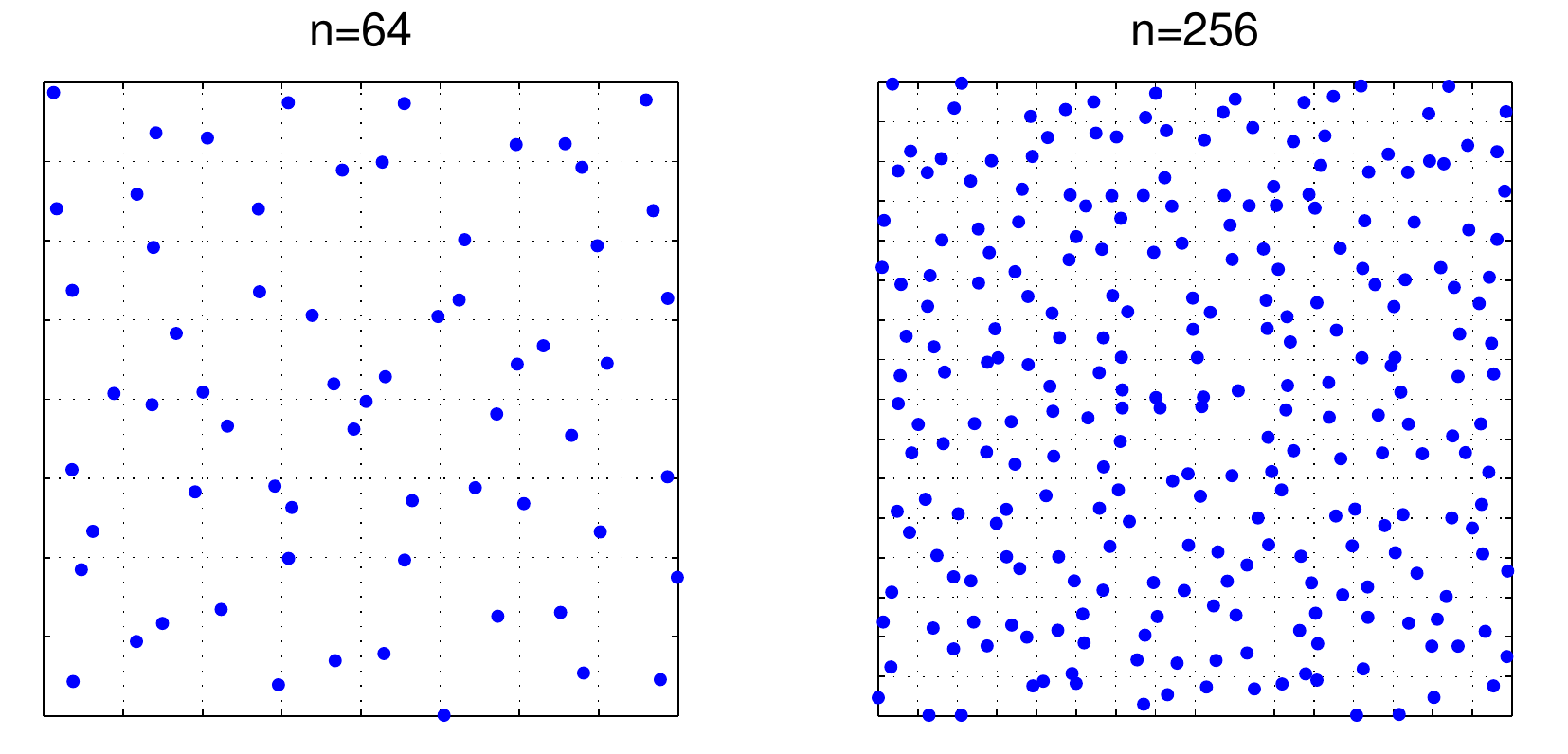}
\caption{
Hilbert mappings into $[0,1]^2$ of $n$ scrambled van der Corput points
in base $2$, for $n\in\{64,256\}$.
}
\label{scrambledvanderCorput}
\end{figure*}

Scrambled net quadrature has a mean squared error of
$O(n^{-3}\log(n)^{d-1})$ for integrands whose mixed
partial derivative taken once with respect to all components
of $x$ is in $L^2[0,1]^d$ \citep{smoovar,localanti}.
The rate in \eqref{boundofsmoothfun2} for $d=1$ is not
as good as that rate even though the algorithms match
in this case.  The explanation is that Lipschitz continuity
is a weaker condition than having the mixed partial in $L^2$.

\subsection{Adaptive sampling}

Integration of discontinuous functions is an important challenge
because there are few good solutions for them.
From the proof of Theorem \ref{ThmLatticeVariance},
we see that intervals of $[0,1]$ in which $f\circ H$ is discontinuous
contribute $O(n^{-1-1/d})$ to the variance, while the other
intervals contribute only  $O(n^{-1-2/d})$. This suggests that we might
improve matters by oversampling the intervals of discontinuity.
In that proof $\tbdy$ collects the indices of $E_i$ touching the boundary of the discontinuity, $\tint$ collects those with $E_i\subset\Omega$ and the ones
contained in $\Omega^c$ don't contribute to the error.
Let us write the estimated integral of the discontinuous function 
$f(x)=g(x)1_{\Omega}(x)$ as
\begin{equation}\label{original}
  \hat{\mu} = \frac{1}{n}\sum_{i\in \tbdy}f(H(x_i))+\frac{1}{n}\sum_{i\not\in\tbdy}f(H(x_i)).
\end{equation}

Suppose we have prior  knowledge about the set $\tbdy$. We could then use
that knowledge to sample $n_0=\ceil{n/\abs{\tbdy}}$ times in each 
stratum $E_i$ for $i\in \tbdy$, and use one sample in the remaining strata as usual. 
From such samples we get the unbiased estimator
\begin{equation}\label{improvedest}
  \hat{\mu} = \frac{1}{nn_0}\sum_{j=1}^{n_0}\sum_{i\in \tbdy}f(H(x_i^{(j)}))+\frac{1}{n}\sum_{i\notin \tbdy}f(H(x_i)),
\end{equation}
where $x_i^{(j)}, x_i\sim \Unif{[\frac{i-1}{n},\frac{i}{n}]}$ independently.
The cost of the estimate (\ref{improvedest}) is at most two times the original estimate (\ref{original}) as it makes at most $2n$ function evaluations.
Roughly half of the evaluations are in the boundary strata.

\begin{thm}\label{thmImproved}
Suppose $f(x)=g(x)1_{\Omega}(x)$ satisfying the conditions of the second
part of Theorem \ref{ThmLatticeVariance}. Then the variance of $\hat{\mu}$ in (\ref{improvedest}) is $O(n^{-1-2/d})$.
\end{thm}
\begin{proof}
From the proof of Theorem \ref{ThmLatticeVariance}, we have $\abs{\tbdy}=O(n^{1-1/d})$ and
$$\bvar{\frac{1}{n}\sum_{i\notin \tbdy}f(H(x_i))}=O(n^{-1-2/d}).$$ 
It remains to bound the variance of first term in the right side of (\ref{improvedest}). 
Similarly to (\ref{var2}), we find that
\begin{align*}
  \bvar{\frac{1}{nn_0}\sum_{j=1}^{n_0}\sum_{i\in \tbdy}f(H(x_i^{(j)}))} &= \frac{1}{n^2n_0^2}\sum_{j=1}^{n_0}\sum_{i\in \tbdy}\var{f(H(x_i^{(j)}))}\\
&=O(n^{-2}n_0^{-1}\abs{\tbdy})=O(n^{-1-2/d}),
\end{align*}
which completes this proof.
\end{proof}

In practice, we have no prior knowledge of $\tbdy$ and
so Theorem~\ref{thmImproved} describes an unusable method.
It does however suggest the possibility of 
adaptive algorithms that both discover and exploit
the presence of boundary intervals.

\section{Numerical Study}\label{experiment}
\subsection{Computational Issue}
In this section we use the image under $H$ of scrambled  van der Corput sampling points on some test integrands with known integrals and compare our observed mean squared errors to the theoretical rates.
We chose the van der Corput points in base $2$
because it is extensible and is easily expressed in base $2$ which conveniently
matches the base used to define the Hilbert curve.
The first step is to randomize the van der Corput sequence using the scrambling scheme of \cite{owen1995}. In the next step, we use the algorithm given by \cite{Butz1971} for mapping the one-dimensional sequence to a $d$-dimensional sequence. 
Butz' algorithm is iterative, requiring a number of iterations equal to the order of the curve, say, $m$. The accuracy of the approximation of each coordinate is $2^{-m}$. 

Using Butz' iteration turns an algorithm with $n$ function values into
one that costs $O(n\log(n))$.
For practical computation, $2^{-m}$ is set to the machine precision, e.g., $m=31$ in our numerical examples, thus the effect is negligible.

Suppose we are going to map a point $x$ in $[0,1]$ to $d$-dimensional point $P$ in $[0,1]^d$, and suppose $x$ is expressed as an $md$-bit binary number:
\begin{equation*}
  x = 0._2\rho_1^1\rho_2^1\cdots\rho_d^1\rho_1^2\rho_2^2\cdots\rho_d^2\rho_1^m\rho_2^m\cdots\rho_d^m.
\end{equation*}
Define $\rho^i=0._2\rho_1^i\rho_2^i\cdots\rho_d^i$. In Butz' algorithm, $\rho^i$ is transformed to $\alpha^i=0._2\alpha_1^i\alpha_2^i\cdots\alpha_d^i$ via some logical operations. See \cite{Butz1971} for details.
The coordinates $p_j$ of $P$ are then given by
\begin{equation*}
  p_j = 0._2\alpha^1_j\alpha^2_j\cdots\alpha^m_j,
\end{equation*}
for $j=1,\dots,d$. To effect this algorithm, one needs to scramble the first $md$ digits of the points in the van der Corput sequence. Suppose that the sample size is $n=2^k$ for integer $k\ge 0$ with $k<md$. At the scrambling stage, we just need to store $n-1$ permutations to scramble the first $k$ digits. The remaining $md-k$ digits are randomly and independently chosen from $\set{0,1}$. When $md\gg k$, the storage requirement of scrambled van der Corput points is much less than that of scrambling a $d$-dimensional digital in base $2$. Note that the Hilbert computations are very fast since they are based on logical operations. 

\subsection{Examples}
We use three integrands of different smoothness to assess the convergence of our quadrature methods:
\begin{itemize}
  \item Smooth function: $f_1(X)= \sum_{i=1}^dX_i$;
  \item Function with cusp: $f_2(X)= \max(\sum_{i=1}^dX_i-\frac{d}{2},0)$;
  \item Discontinuous function: $f_3(X)= 1_{\set{\sum_{i=1}^dX_i>\frac{d}{2}}}(X)$.
\end{itemize}
Note that $f_1$ and $f_2$ are Lipschitz continuous. From Theorem \ref{ThmvanderCorputVariance}, the theoretical rate of mean squared error
for these two function is $O(n^{-1-2/d})$.
The corresponding rate for $f_3$ is $O(n^{-1-1/d})$ as its discontinuity boundary is 
has finite Minkowski content. Figure \ref{fig1} shows the convergence graphs for $d=2,3,8,16$. These results support the theoretical rates shown in Theorem \ref{ThmvanderCorputVariance}.

\begin{figure}[t]
  \centering
\includegraphics[width=\hsize]{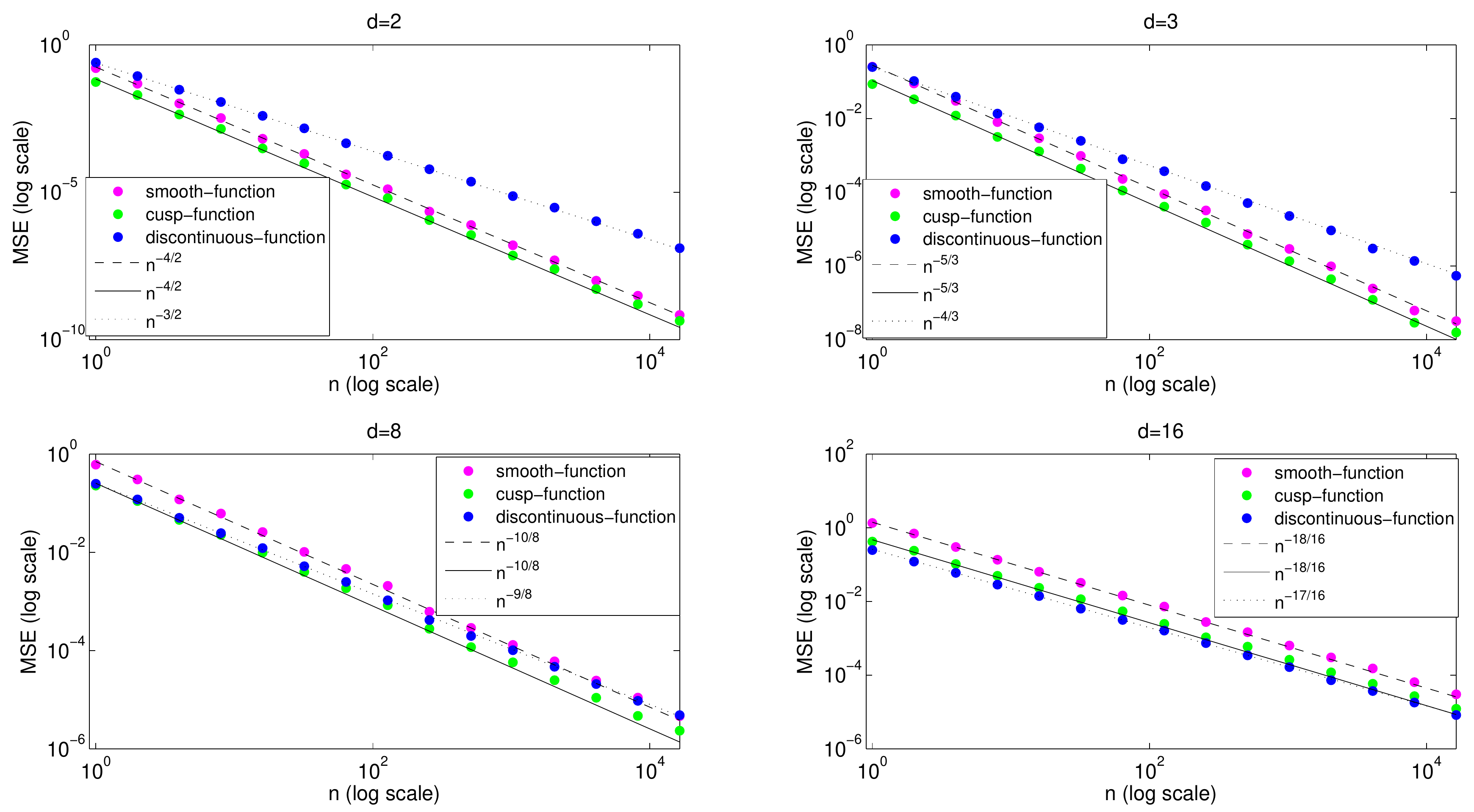}
  \caption{
MSE versus $n$ for functions 
$f_1$ (smooth)  $f_2$ (cusp) and $f_3$ (discontinuous), 
in dimensions $d=2,3,8,16$. The sample points are first $n$
van der Corput points (scrambled) 
for $n=2^k,k=0,\dots,14$. 
The reference lines are proportional to labeled rates, which reflect the theoretical rates. 
The MSEs are calculated based on $1000$ repetitions.
}\label{fig1}
\end{figure}

\section{Discussion}\label{conclusion}

In this paper, we study a quadrature method combining the
one-dimensional QMC points with the Hilbert curve in dimension $d$. We
find that the star-discrepancy has a very poor convergence rate
$O(n^{-1/d})$ in $d$ dimensions, which is the rate one would attain
by sampling on an $m^d$ grid.
Although this rate seems slow, 
deterministic quadrature for Lipschitz functions, using $n=m^d$
points in $[0,1]^d$  has an error rate of $O(n^{-1/d})$
(and a lower bound at that rate) according to \cite{sukharev1979optimal} as 
reported in \cite{novak1988deterministic}.
See also \cite{sobo:1989}. When $f$ has bounded variation on $[0,1]^d$,
then the Hilbert mapping of a low discrepancy point set in $[0,1]$
attains the optimal rate.

Randomized van der Corput sampling has a mean squared
error of $O(n^{-1-2/d})$ for Lipschitz continuous functions.
This is the same rate seen for samples of size $n=m^d$
in the stratified sampling method of 
\cite{dupa:1956} and \cite{haber:1966}, which takes one
or more points independently in each of $m^d$ congruent
subcubes of $[0,1]^d$.
Compared to $O(n^{-1/d})$, this rate
reflects the widely seen error reduction by $O(n^{-1/2})$
commonly seen in the randomized setting versus worst
case settings.

Both deterministic and randomized versions of this Hilbert
space sampling match the rates seen on grids, without requiring
$n$ to be of the form $m^d$. This is why we think of the
Hilbert mapping of van der Corput sequences as extensible
grids.

The figures in~\cite{gerb:chop:2014} show a decreasing
rate improvement over Monte Carlo  as the dimension
of their examples increases.
Our results do not yield a convergence rate for their
algorithm. They use
the inverse $h_m:[0,1]^d\to[0,1]$ of the Hilbert function $H_m$,
in addition to $H_m$.
The function $H$ is not invertible as there is a set
of measure $0$ in $[0,1]^d$ whose points have more
than one pre-image in $[0,1]$. 
For large $m$, the function $h_m$ is very non-smooth and has enormous
variation because nearby points in $[0,1]^d$ can arise
as the images under $H_m$ of widely separated points in $[0,1]$.

Our main theorems \ref{ThmDiscrepancy}, \ref{ThmDiscrepancyVDC},
\ref{ThmLatticeVariance}, and \ref{ThmvanderCorputVariance}
on star discrepancy and sampling variance are not strongly tied
to the Hilbert space-filling curve. The space-filling curves of Peano
and Sierpinski also satisfy the H\"older inequality with exponent $1/d$
that we based our arguments on, although with a different constant.
As a result, the same rates of convergence hold for stratified and
van der Corput sampling along these curves.
The Lebesgue space-filling curve, also called the $Z$ curve, differs
from the aforementioned curves in that it is differentiable almost
everywhere.  It also satisfies H\"older continuity, but the exponent
is $\log(2)/(d\log(3))$ which is  worse than $1/d$ that holds
for the other curves.  Using the Lebesgue curve is roughly
like multiplying the dimension by $\log_2(3)$, compared to using
the Hilbert curve.
See~\citet[Chapter 4]{zumbusch2003parallel} for these properties
of space-filling curves.

\section*{Acknowledgments}

We thank Erich Novak for helpful comments.
ABO was supported by the US NSF under grant DMS-0906056.
ZH was supported by a PhD Short-Term Visiting abroad Scholarship of Tsinghua University.

\bibliographystyle{apalike}
\bibliography{myBiBLibray}
\end{document}